%% file: main.tex
\documentclass[sigconf, nonacm]{acmart}

\input{packages}
\input{macros}

\crefname{observation}{Observation}{Observations}
\crefname{corollary}{Corollary}{Corollaries}

\begin{document}
\title{Determinacy of Real Conjunctive Queries. The Boolean Case.}
 
 \author{Jarosław Kwiecień}
 \affiliation{%
    \institution{University of Wrocław}
    \country{}
    \streetaddress{}
    \city{}
    \state{}
    \postcode{}}
 \author{Jerzy Marcinkowski}
 \affiliation{%
    \institution{University of Wrocław}
    \country{}
    \streetaddress{}
    \city{}
    \state{}
    \postcode{}}
 \author{Piotr Ostropolski-Nalewaja}
 \affiliation{%
    \institution{University of Wrocław}
    \country{}
    \streetaddress{}
    \city{}
    \state{}
    \postcode{}}

\begin{abstract}
In their classical 1993 paper \cite{CV93} Chaudhuri and Vardi notice that some fundamental database theory
results and techniques fail to survive when we try to see query answers as
bags (multisets) of tuples rather than as sets of tuples.

But disappointingly,  almost 30 years after \cite{CV93},
the bag-semantics based database theory is still in the infancy.
We do not even know whether conjunctive query containment is decidable.
And this is not due to lack of interest, but because, in the multiset world,
everything suddenly gets discouragingly complicated.

In this paper we try to re-examine, in the bag semantics scenario, the query determinacy problem, which has recently been intensively studied in the set semantics scenario. We show that query determinacy (under bag semantics) is decidable for boolean conjunctive queries and undecidable for unions of such queries (in contrast to the set semantics scenario, where the UCQ case remains decidable even for unary queries). We also show that -- surprisingly -- for path queries determinacy under bag semantics coincides with determinacy under set semantics (and thus it is decidable).
\end{abstract}

\maketitle

\section{Introduction}

\subsection{The context}\label{intro-1}

This paper is about the query determinacy problem. So let us maybe start with a definition:

\begin{definition}\label{pierwsza}
\textbullet~ For a query $\query $  and a finite set of queries $\views $, we say that $\views$ determines $\query$ (denoted as $\views \detttt \query$) 
if the implication:
\begin{equation*}\label{eq:determinacy}\tag{\ding{171}}
(\forall v\in V \;\;   v(\db) = v(\db')) \implies \query(\db) = \query(\db')
\end{equation*}
holds for every pair $\db, \db'$ of finite\footnote{Both ``finite'' and ``unrestricted'' versions of this problem were considered, but in this paper let us concentrate on the finite one, which is the only one to make sense in the multiset scenario.} structures\footnote{Where $v(D)$ is the result of applying $v$ to $D$. }.

\textbullet~  An instance of the determinacy problem, for a query language 
$\querylang$, consists of a query $\query \in \querylang$ and a finite set of views $\views \subseteq \querylang$. We ask whether 
$\views\detttt \query$.
\end{definition}

Many different variants of the determinacy  problem, for various query languages, and (when applicable) various arities of queries, have been
studied in the last three decades. And  the point has been reached, where we have a pretty complete classification
of the variants, in the sense that we know which of them are decidable (few) and which are not (most).

So, for example,  as observed in \cite{M20}, the problem is decidable if the queries in $V$ are unary  UCQs\footnote{``Unary'' means that they have one free variable. Similar result for unary conjunctive queries was proven 
 in \cite{NSV07}.}
 (unions of conjunctive queries) and $\query$ is any UCQ. 
 Let us outline how one can prove this:
 
 As noticed in he paper \cite{GM15}, 
  $V \detttt \query$ holds if and only if:\\
 \hspace*{15mm} (*) $\;\;Chase(TGD(V),green(\query))\models red(\query)$\\
  where $green(\query)$ and $red(\query)$ are some structures that can easily be constructed from $\query$ and $TGD(V)$ is some set of Tuple Generating Dependencies
  which can easily be constructed from $V$, and where $Chase(TGD(V),green(\query))$ is a result of applying the TGDs from $TGD(V)$ to $green(\query)$ until the fixpoint is reached. Then, it is easy to see that if the queries from $V$ are unary then the TGDs in $TGD(V)$ are frontier one.
  And  query entailment\footnote{That is, condition (*) above.} is decidable for sets of such TGDs \cite{MLMS11}. Then, if one is unhappy 
  with the fact that  $Chase(TGD(V),green(\query))$ is potentially  
  infinite, leading to infinite  $D$ and $D'$, 
  the finite controlabillity result for frontier-one TGDs (implied by \cite{BCS15}) 
  can be used to replace $Chase(TGD(V),green(\query))$ with a finite structure
  with the desired properties.
  
  We do not really want our readers to understand this complicated
  reasoning (unless they already do). We only outline it in order to show that
 database theory has  reached the point where it is no  longer merely a 
set\footnote{Or maybe ``multiset'' would be a better term in this context, as some of the results
were produced more than once.} of  results about the fundamental 
notions and phenomena,
but  a real scientific theory, able to explain and interpret facts which are apparently totally unrelated: we do not believe that the authors of 
\cite{MLMS11} and 
\cite{BCS15} ever expected their results to be used in a decidability proof of a variant of the determinacy problem.

  
  Unfortunately, this beautiful palace of database theory, both the results and the tools,
   collapses like a house of cards when we try to be slightly more realistic and assume that the queries do not return {\bf sets} of tuples, but they return {\bf multisets} (or {\em bags}) of tuples.

  
  And this is not a new observation.
  It was already spotted in \cite{CV93}, 
  where the authors try to see what happens 
  to the most important database theory fundamental, query containment, if bag semantics is assumed, concluding that the 
  {\em 
 ``techniques from the set-theoretic
 setting do not carry over to the
 bag-theoretic setting''.}
 
 The paper \cite{CV93} was understood, at least by part of the community, as a call {\em ``for a {\bf re-examination} of
the foundations of databases where the fundamental concepts and algorithmic problems are
investigated under bag semantics, instead of set semantics''} (see \cite{AK21}, page 2). But only rather limited progress has been achieved.
Even the decidability of conjunctive query containment
problem remains open in the multiset world. 
And this is in spite of a considerable effort,
which is reflected by a list of publications.

First, in 1995 \cite{IR95} show that containment of 
  UCQs, which is in NP when the classical (set) semantics is considered, becomes undecidable for the multiset semantics. 
Then (among other papers) there are \cite{JKV06} where it is shown that containment is undecidable if 
inequalities are allowed in conjunctive queries and \cite{ADG10} which shows decidability (and establishes complexity) for several
simple subcases.  And then, finally, there is a  paper \cite{AKNS20},
where query containment is in an elegant way related to the information-theoretic notion of entropy, and it is shown that 
decidability of even a quite limited subproblem of query containment would imply a solution to a long standing open problem in information theory.

Apart from the line of research focused on the query containment problem, the number of such {\em re-examination} attempts, while growing, remains low. And this 
is -- we understand -- not because of lack of interest, but because (as the containment problem illustrates) 
everything suddenly gets very complicated 
when multiset semantics is assumed. One example we know about is the recent paper 
\cite{AK21} where the authors re-examine the old result from 
\cite{BFMY83}, that a database schema is
acyclic if and only if the local-to-global consistency property for relations over that schema
holds.

  \subsection{Our contribution (and the future work).}

In this paper we attempt {\em a re-examination}, under multiset semantics, of the query determinacy problem. 

This means that we now read the equalities in formula \ding{171} 
as
equalities of multisets. To distinguish we will use the symbol 
$\detset$ to denote the old style set-semantics determinacy and $\detbag$ for 
determinacy under multi-set semantics.

  The first question one naturally needs to ask here is whether $\detbag$ is really a different notion
  than $\detset$. And, if they are indeed different,  the second question is:
  does at least one implication hold? Like in the case
  of query containment, where, as noticed already in \cite{CV93},
 containment under multiset semantics  is a strictly stronger property than containment under set semantics?
 
 To show that the two versions are really different let us use:
 
 \begin{example}\label{ex-1}
 Let $\query$ be the query $\exists u,y,z \; P(u,\mathbf{x}),R(\mathbf{x},y), S(y,z)$
 and let $V$ consist of two conjunctive queries:
 $$\exists u,y \;P(u,\mathbf{x}),R(\mathbf{x},y) \;\;\;\;\;\;\;\;\;\;\;\;\; \exists  y,z \: R(\mathbf{x},y), S(y,z)$$
 Then is is easy to see that $V\detset \query$ but $V\not \detbag \query$.
  \end{example}

 Regarding the second question notice that while equality 
 of $\psi(D)$ and $\psi(D')$ (for some query $\psi$) under multiset semantics  implies that they are 
 also equal under set semantics,
 the formula \ding{171}  has both positive and negative occurrence of equality of the answer sets. 
 So it is not obvious at all that multiset determinacy always implies determinacy in the set semantics world. 
 And indeed:
  
  \begin{example}\label{ex-zly} Let $\query$ be the query $\exists x \; R(x)$ and let $V$ consist of two queries:\vspace{-4mm}
  $$\;\;\;\;v_1=\exists x \; P(x) \;\;\;\;\;\;\;\;\;\;\;\;\; v_2=\exists x \; P(x)\vee \exists x \; R(x)$$
  
  Then it is easy to see that $V\not\detset \query$. But under 
  the multiset semantics for each $D$ we have $\query(D)=v_2(D)-v_1(D)$ (since we consider boolean queries here, the 
  answers are natural numbers), which implies that  $V\detbag \query$.
  \end{example}
  
  Can such example be constructed for conjunctive queries rather than UCQs? We do not know. We conjecture that the answer is ``no'',
  but proving it will probably be hard. What we can show (and we find it a bit surprising, because the situations where 
  set-semantics based notions coincide with their multiset-semantics counterparts seem to be rare)  is:
  
  \begin{theorem}\label{th-1} If $V$ is a set of path queries, and $\query$ is a path query\footnote{For a definition of path queries see \cref{sec:path-queries}} then 
   $V\detset \query$ if and only if $V \detbag \query$.
  \end{theorem}
  
  Determinacy of path queries (under the set semantics) is one of the few decidable cases \cite{A11}, and, as Theorem \ref{th-1} implies,
  it remains decidable in the multiset semantics world. For the proof of Theorem \ref{th-1}  see \cref{sec:path-queries}. Notice also that the queries from Example \ref{ex-1} are not far from being path queries, but still,
  for some reason, the thesis of Theorem \ref{th-1} does not hold for them.
  
 But the main focus of this paper is on understanding query determinacy  in the case of boolean queries.
 We first present:
 
 \begin{theorem}\label{th-2} The problem whether, for a set 
 $V$ of boolean UCQs and for another boolean UCQ $\query$, it is true that  $V\detbag \query$, is undecidable. 
  \end{theorem}
  
This is in stark contrast to the situation in the set semantics world where, 
as we already mentioned in Section \ref{intro-1}, 
determinacy is decidable even for unary UCQs, not just boolean. 
But the proof of Theorem 2 is not hard. In order to show it, it was enough to notice that the 
``$p_1 \vee p_2$ trick'' from \cite{SV05} (or the ``cold-hot'' trick from \cite{M20}) can be safely used in the multiset semantics world. And then  to reuse the Hilbert 10th problem encoding from \cite{IR95}.
  
  Finally, our main technical contribution is: 
  
  \begin{theorem}\label{th-3}\label{theorem-3} The problem whether, for a set 
 $V$ of boolean CQs and for another boolean CQ $\query$, it is true that $V\detbag \query$, is decidable. 
  \end{theorem}
  
 The proof of Theorem \ref{th-3} is
 presented in Sections \ref{sec-main-th-1}-\ref{usinggoodS-proof}. A corollary from the proof of Theorem \ref{th-3}  is that 
 for boolean conjunctive queries $\detbag$ is a strictly stronger property than  $\detset$.\vspace{-2mm}\\
 
 \noindent {\bf Future work.} The natural open question we leave is the decidability status of the 
 CQ determinacy  for the multiset semantics, that is of the problem whether, for a set 
 $V$ of  CQs (with free variables) and for another CQ $\query$, it is true that $V\detbag \query$. The
 encoding method from the proof of  Theorem \ref{th-2} is useless when disjunction is no longer available. 
 And also the techniques from the proof of Theorem \ref{theorem-3} do not seem to generalize to the 
 scenario with free variables.

 
 
 \subsection{The tools. And related works.}
 
Regarding {\bf \small the tools}  used in the proofs of Theorem \ref{th-1} and  Theorem \ref{th-3},
let us quote \cite{CV93} again: {\em techniques  from the set-theoretic setting}
 {\em do not carry over to the bag-theoretic setting}. 
 The green-red chase (mentioned in Section \ref{intro-1}), which is a fundamental tool to study 
 determinacy  in the set-semantics world,
 just vanishes in the multiset setting, together with the results that depend on it, 
 like undecidability  for the CQ case.
 And in general, the importance of concepts that stem from 
the first order logic diminishes in this new world. Instead, tools based on notions from linear algebra arise 
in a very natural way here. This is not at all a new observation: in order to read \cite{AK21} one also needs to dust off the 
linear algebra textbook.

While we are (as far as we know) the first to consider query determinacy under multiset semantics,
there exists a line of research in database theory which concentrates on the number of 
answers to a query (homomorphisms), including paper \cite{CM17} (again in a natural way using arguments from linear algebra).
And also, due to solely mathematical motivations, such  homomorphism counts (and some related numbers) were studied 
by researchers in combinatorics, with numerous  papers published, including \cite{lovasz} and \cite{ELS79}, and a book\footnote{We only had access to a free version of this book available on the web.}
\cite{L19}.
Some of the results regarding  homomorphism counts are useful for us (see Section \ref{good-s-proof} where we use the main result from 
\cite{lovasz}). Some, while not directly useful, are  related to our paper, for example there is a construction 
in \cite{CM17} resembling Step 1 (and partially also Step 2) from our construction in Section \ref{good-s-proof}.

The title of \cite{ELS79} may suggest that there is a connection to determinacy and (as we learned)
some less careful readers can have an impression that the main result from \cite{ELS79} 
is almost our Theorem \ref{theorem-3}. So
let us take some space here to explain why this is not the case\footnote{It may be a good idea to skip the rest of this Section now and come back here after you read 
Section \ref{sec-main-th-1} at the earliest.}.

 
A set of connected non-isomorphic graphs ${\mathcal H} =\{H_1,\ldots H_m \}$ is  consideerd in \cite{ELS79}. For
$H\in \mathcal H$ and for another graph $G$
the number $t(H,G)$ (,,homomorphism density'') is defined as the probability that a random mapping from 
the set of verticies of $H$ to the set of verticies of $G$ will be a homomorphism. 

Let now $S$ be the set  $\{\langle t(H_1,G), \ldots,t(H_m,G) \rangle : G \text{ is a graph} \}$,
which clearly is a subset  of $[0,1]^m$  or, to be more precise, of $({\mathbb Q}\cap [0,1])^m$.
The main theorem 
of \cite{ELS79} (Theorem 1 there)
says that:

\noindent
{\bf (*)} \hspace{8mm} $S$ contains a subset $B$  which is dense in some ball.

Then, it seems to be claimed\footnote{Remarks after Corollary 5.45 in  \cite{L19}; unfortunately 
the language is quite sloppy there, and it is not entitely clear for us how this part of text should correctly be interpreted.}
 in \cite{L19} that it follows from (*) that no functional dependence
between the numbers  $ t(H_1,G),\ldots,t(H_m,G) $ can exist, meaning that 
$t(H_m,G)$ cannot be a function of arguments $t(H_1,G),\ldots ,t(H_{m-1},G)$. In our language this would mean 
that:

\noindent
{\bf(**)} \hspace{4mm}  $t(H_1,G),\ldots ,t(H_{m-1},G)$ do not determine $t(H_m,G)$ .

If this was indeed true that (*) implied (**) then one could use 
 the graph blow-up technique from \cite{L19} (Theorem 5.32) to translate the language of ,,homomorphism densities''
 into the language of homomorphism counts and, as a result, prove our Corollary \ref{wnioseczek}, which is a very special case of our Theorem \ref{theorem-3}. 
 
 But (*)  {\bf does not} imply (**). Let us define $C$ as the projection of $S$ on the first m-1 coordinates,
that is $C= \{ \langle q_1,\ldots,q_{m-1} \rangle : \exists q_m  \; \langle q_1,...,q_{m-1}, q_m \rangle \in S\}$.

Then (**) means that $S$ cannot be the graph of some function $f: C \rightarrow \mathbb Q$.
But all (*) tells us about $S$
is that its topological closure contains a ball.
And it is  easy  to construct such a  function
 $f:([0,1]\cap {\mathbb Q})^{m-1} \rightarrow {\mathbb Q}\cap [0,1]$ 
that the topological closure of the graph of $f$ not only contains a ball
but is actually the entire cube $[0,1]^m$.

What does indeed follow from (*) is that no such {\bf \small continuous} function $f$ can exist,
so in particular $t(H_m,G)$ cannot be expressed from $t(H_1,G), \ldots ,t(H_{m-1},G)$
by operations which preserve
continuousness.
But then it is a completely different story, as continuousness
may make sense
when talking about homomorphisms density, but not in the context of
homomorphism count.



\section{Preliminaries}\label{sec:prelims}

\subsection{ Database Theory Notions}

\subsubsection*{\bf Multisets} A multiset $X$ is a mapping $Y \to \mathbb{N}$ where $Y$ is some specified set\footnote{We (of course) think that $0\in \mathbb N $.}. With $X[a]$ we will denote the number of occurrences of $a$ in $X$. We write that $X[a] = 0$ if $a \not\in Y$. A {\em union} $X \cup X'$ of two multisets $X$ and $X'$ is a multiset such that $(X \cup X')[a] = X[a] + X'[a]$. We define other multiset operators analogously.

\subsubsection*{\bf Structures}
A {\em schema} $\schema$ is a finite set of relational symbols. A schema $\schema$ is {\em $n$-ary} if an arity of its relations is at most $n$.
A {\em  structure (or database)} $\db$ over schema $\schema$ is a finite set\footnote{Which means that we assume that answers to the queries are multisets, but the structures  are sets. However, all our results and techniques would survive if we defined structures which are multisets of facts.} consisting of facts. A {\em fact} is simply an atom $A(\Vt)$ where $\Vt$ is a tuple of {\em terms} from some fixed  infinite set of constants. 
The {\em active domain} of $\db$ (denoted with $\domain{\db}$) is the set of constants that appear in facts of $\db$.

\subsubsection*{\bf Homomorphisms}
For two structures $\structure$ and $\structure'$ over schema $\schema$, a homomorphism from $\structure$ to $\structure'$ is a function $h: dom(\structure) \to dom(\structure')$ such that for each atom $A(\Vt) \in \structure$ it holds that $A(h(\Vt)) \in \structure'$.
A set of homomorphisms from $\db$ to $\db'$ is denoted with $\homs{\db, \db'}$.
Note, that $\nhoms{\emptyset, \db} = 1$ for the empty structure $\emptyset$.

\subsubsection*{\bf Conjunctive Queries (CQs).}
A {\em conjunctive query} $\Phi = \exists{\Vy}\;\phi(\Vx, \Vy)$ is a first order formula such that $\phi(\Vx, \Vy)$ is a conjunction of atoms over variables from $\Vx$ and $\Vy$. With $vars(\Phi)$ we will denote the set of variables of $\phi(\Vx, \Vy)$. The {\em arity} of CQ $\Phi$ is simply $|\Vx|$.

The {\em frozen body} of a CQ $\Phi = \exists{\Vy} \; \phi(\Vx, \Vy)$ is a structure obtained from $\phi(\Vx, \Vy)$ by bijective replacement of variables with fresh constants.
For a CQ $\Phi = \exists{\Vy}\;\phi(\Vx, \Vy)$ and a structure $\db$, with $\homs{\Phi, \db}$ we denote the set of all homomorphisms from the frozen body of $\Phi$ to $\db$. A {\em result} $\Phi(\db)$ of a CQ $\Phi$ over a structure $\db$ is a multiset such that $\Phi(\db)[\Va] = |\set{h \in \homs{\Phi, \db} \;:\; \Va = h(\Vx)}|$.

\subsubsection*{\bf Path Queries}
For a binary schema $\schema$ a {\em path query} $\Lambda$ is a CQ of the  form
$\exists{x_1,\ldots, x_{n-1}} \;\; R_1(x,x_1), R_2(x_1, x_2), \ldots, R_n(x_{n-1}, y).$

Let $\schema^*$ denote the set of all words over relational symbols from $\schema$. Given the nature of path queries we will identify them with words from $\schema^*$, so instead of writing $\Lambda = \exists{x_1, x_2} A(x, x_1), B(x_1, x_2), C(x_2, y)$ we may conveniently write\footnote{Note however, that an empty word $\varepsilon$ is identified with the query $\Lambda(x, y) = "x = y"$, although it is not a valid path query.} $\Lambda = ABC$.

\subsubsection*{\bf Boolean Queries}
A CQ  $\query$ with no free variables is called {\em boolean}.
 Boolean CQs will be always identified with their frozen bodies.

Accordingly to previous definitions a result $\query(\db)$ of a boolean CQ  $\query$ over some structure $\db$ is a multiset containing $\nhoms{\query, \db}$ copies of the empty tuple. For brevity we write $q(D)$ instead of $q(D)[\tuple{}]$, so $q(D) = |hom(q, D)|$.

A {\em union of boolean conjunctive queries (boolean UCQ)} $\Psi$ is a disjunction of a finite number of boolean CQs. 
A {\em result} $\Psi(\db)$ of a boolean UCQ $\Psi$ over a $\db$ is the natural number $ \sum_{\Phi \in \Psi} \Phi(\db)$.

A boolean CQ
$\query$ is {\bf  contained under set semantics} in a boolean CQ $\query'$ 
(denoted as  $\query \subseteq_\sett \query'$)
if for every structure $\db$ it holds that 
$\query(\db) > 0 \Rightarrow \query'(\db) > 0$. 
It is well-known that $\query \subseteq_\sett \query'$ if and only if $\homs{\query', \query}$ is non-empty.

\subsection{\bf Graph Theoretic Tools}\label{gtt}


\subsubsection*{\bf Operations on Structures}
Following \cite{lovasz} we will use some operations on structures. For structures $A$ and $B$  over schema $\schema$:

\noindent
\textbullet~  $A + B$ is a disjoint union\footnote{That is if $dom(A) \cap dom(B) \neq \emptyset$ we 
bijectively rename variables of $B$ with fresh ones and then make $A + B = A \cup B$} 
of $A$ and $B$;
    
    \noindent
\textbullet~     $A \times B$ is a structure such that $dom(A \times B) = dom(A) \times dom(B)$ and for any $R \in \schema$ the following holds: $R(\pair{a_1, b_1}, \ldots, \pair{a_k, b_k})$ is an atom of $A \times B$ if and only if $R(a_1, \ldots, a_k) \in A$ and $R(b_1, \ldots, b_k) \in B$;
    
    \noindent
 \textbullet~we use  symbols $\sum$ and $\prod$  as generalized $+$ and $\times$ in the usual way;
  
  \noindent  
\textbullet~     for $t \in \N_+,\; tA = \sum_{i=1}^t A$ and $A^t = \prod_{i=1}^t A$. Furthermore, $0A$ is an empty structure and $A^0$ is a singleton $\{ \alpha \}$ such that for any $R \in \schema$ $R(\alpha, \alpha, \ldots, \alpha) \in A^0$ ($\alpha$ has loops of all types).

\subsubsection*{\bf Graph Theoretic Lemma}
From~\cite{lovasz} we recall:
\begin{lemma}\label{lem:hom-properties}
Let $A, B, C$ be structures and $t\in \N$, then:
\begin{enumerate}
    \item If $A$ is connected, then $\nhoms{A, B+C} = \nhoms{A, B} + \nhoms{A, C}$
    \item If $A$ is connected, then $\nhoms{A, tB} = t \cdot \nhoms{A, B}$
    \item $\nhoms{A, B \times C} = \nhoms{A, B} \cdot \nhoms{A, C}$
    \item $\nhoms{A, B^t} = \nhoms{A, B}^t$ 
    \item \label{eq:hom-properties-5} $\nhoms{A+B, C} = \nhoms{A, C} \cdot \nhoms{B, C}$
\end{enumerate}
\end{lemma}

%
%

\subsection{Basic Mathematical Tools and Notations}

We are going use standard notation from linear algebra, which should be clear in most cases. Below we describe all the conventions that might be non-obvious:

\noindent
   {\large\textbullet~} For a set $A \subseteq \R^k$, $\linspan(A)$ means the linear span of $A$ (i.e. the smallest linear space containing $A$). For a set $B \subseteq \R$, we define $\linspan^B(A)= \{ b_1 \vec{a_1} + ... + b_n \vec{a_n} \mid n \in \N; \vec{a_1}, ..., \vec{a_n} \in A; b_1, ..., b_n \in B \}.$
    
    \noindent
   {\large\textbullet~} For two vectors $\vec u, \vec{u'} \in \R^k$, $\dotprod{\vec u}{\vec {u'}}$ denotes the dot product of $\vec u, \vec {u'}$. Vector $\vec u$ is \emph{orthogonal} to $\vec {u'}$ if and only if $\dotprod{\vec u}{\vec {u'}} = 0$.
   
 \noindent
   {\large\textbullet~} For a vector $\vec u \in \R^k, i \in \range{k}$, $u(i)$ denotes the value of the $i$-th coordinate of $\vec u$.
    
    \noindent
   {\large\textbullet~} For a matrix $M \in \R^{k \times k}, i, j \in \range{k}$, $M(i, j)$ denotes the value of the element in the $i$-th row and $j$-th column of $M$.
   
    \noindent
   {\large\textbullet~} For a matrix $M \in \R^{k \times k}$ and a set $A \subseteq \R^k$, $M(A) = \{ M \vec x \mid \vec x \in A \}$.\\

\vspace{-2.5mm}
We will use the following well-known mathematical facts:
\begin{fact}\label{fact:orthogonal-vector}
Let $\vec{u}_1, ..., \vec{u}_n, \vec{u} \in \Q^k$ such that $\vec{u} \notin \linspan\{ \vec{u}_1, ..., \vec{u}_n \}$. Then there is a vector $\vec{z} \in \Q^k$ such that $\vec{z}$ is orthogonal to $\vec{u}_1, ..., \vec{u}_n$ but is not orthogonal to $\vec{u}$.
\end{fact}
\begin{fact}\label{fact:homeomorphism}
If matrix $M \in \R^{k \times k}$ is nonsingular, then the mapping $\vec x \mapsto M \vec x$ is a homeomorphism (a continuous bijection whose inverse function is continuous too).
\end{fact}
\begin{fact}\label{fact:q-dense}
$\Q^k$ is a dense subset of $\R^k$, i. e., for any $\vec x \in \R^k$ and $r > 0$ there is $\vec y \in \Q^k$ such that $\|\vec x - \vec y\| < r$.
\end{fact}
\begin{corollary}\label{cor:interior}
Suppose $M \in \R^{k \times k}$ is nonsingular. Then there is $\vec p \in M(\R_{\geq 0}^k) \cap \Q^k$ such that
\begin{equation*}\label{eq:interior}\tag{$\star$}
\exists r > 0 ~ \forall \vec x \in \R^k \;\; \| \vec x - \vec p \| < r \Rightarrow \vec x \in M(\R_{\geq 0}^k)
\end{equation*}
\end{corollary}
\noindent
{\sc Proof.}
From  \cref{fact:homeomorphism} we know that the set $M(\R_{\geq 0}^k)$ has non-empty interior (i. e. the set of points satisfying \ref{eq:interior}), since it is a homeomorphic image of a set with non-empty interior. By \cref{fact:q-dense} we get that this interior must contain a point with rational coordinates.\hfill $\square$

\subsubsection*{\bf Important Notational Convention.} $0^0$ equals $1$ in this paper.

\section{The Path Queries Case}\label{sec:path-queries}
\setcounter{theorem}{0}

In this section we prove:

\setcounter{theorem}{0}
\begin{theorem} If $V$ is a set of path queries, and $\query$ is a path query, then 
   $V\detset \query$ if and only if $V \detbag \query$.
  \end{theorem}

One can find this theorem a bit surprising. Path queries are a reasonably wide class of queries. And we have already learned
that one should not expect a set-semantics based notion to agree with its multi-set based counterpart on 
a wide class of objects\footnote{On the other hand, for path queries, query containment under set semantics also (trivially)
coincides with query containment under bag semantics. We have no idea whether there is any relation between this observation and Theorem \ref{th-1}.}. But, as it turns out, both versions of determinacy for path queries enjoy the same elegant combinatorial characterisation:

\begin{definition}
For a set $V$ of path queries and for another path query $q$ we define an undirected graph $G_{q,V}$ as follows:
\begin{itemize}
    \item $dom(G_{q,V}) = \{ w \in \schema^* \mid w \text{ is a prefix of } q \}$.
 In particular, the empty word $\varepsilon$ and $q$ itself\footnote{Recall that we identify path queries with words over alphabet $\schema$.} are elements of $dom(G_{q, V})$.
    \item There is an edge between $w$ and $w'$ if and only if $w'=wv$ for some $v\in V$.
\end{itemize}
\end{definition}

The following fact is well known \cite{A11,M20}:

\begin{fact}
$V \detset q$ iff there is a path in $G_{q, V}$ from $\varepsilon$ to $q$.
\end{fact}

In order to prove Theorem \ref{th-1} we will show that the same is true for determinacy in the multiset setting:

\begin{lemma}\label{lem:path-queries}
$V \detbag q$ iff there is a path, in $G_{q, V}$, from $\varepsilon$ to $q$.
\end{lemma}

\noindent
{\bf  The rest of this section is devoted to the proof  of Lemma \ref{lem:path-queries}.}

It turns out that the (simple) proof of the  ($\Rightarrow$) direction for the set semantics survives also in the multiset context.  We include it here for the sake of completeness, but due to the space limitations defer it to  Appendix B.\\

\noindent
{\bf Let us now deal with the  ($\Leftarrow$) direction.}
Assume that $q$ and $V$ are fixed and such that there is a path in $G_{q,V}$, of some length $m$, 
from $\varepsilon$ to $q$.
This means that there exist:\\
\textbullet~ a sequence $w_0, w_1,...w_m$ of prefixes of $q$, with $\varepsilon=w_0$ and $w_m=q$;\\
\textbullet~ a sequence of numbers $\epsilon_j$, for $j\in \{1,\ldots m\}$, each of them either equal 1 or $-1$;\\
\textbullet~ a sequence $v_{p_1}, \ldots v_{p_m}$ of elements of $V$
such that, for each $j\in \{1, \ldots m\}$, one of the conditions is true:

\hspace*{3mm}\textbullet~ $\epsilon_j=1$ and $w_j=w_{j-1}v_{p_j}$;\hfill
\textbullet~ $\epsilon_j=-1$ and $w_jv_{p_j}=w_{j-1}$.\\

\vspace{-1mm}
\noindent
We are going to show that in such case there will also be $V \detbag q$.
So we assume that there are two stuctures $D$ and $D'$ such that $v(D)=v(D')$ for each $v\in V$. 
Without loss of generality we can also assume that domains of $D$ and $D'$ are 
equal\footnote{By domain we do not mean the active domain here: we accept that there are elements, in $dom(D)$
or $dom(D')$ which do not appear in any facts.}, so let
$dom(D)=dom(D')= \{ a_1, ..., a_n \}$.

\subsection{The $q$-walks and  how to reduce them.}

Let $\bar\schema=\schema \cup \schema^{-1}$ be a new alphabet\footnote{Or schema, in the world of
path queries words and queries are the same thing.}, where $\schema^{-1}=\{R^{-1}\mid R\in\schema\}$. 

\begin{definition}\label{def:walk}
Let $w \in \bar\schema^*$ and $w = A_1^{\iota_1} ... A_k^{\iota_k}$ for some
$A_1, ..., A_k \in \schema$ and for some $ \iota_1, ..., \iota_k \in \{ 1, -1 \}$. 
Then $w$ is called a \emph{$q$-walk } if:

\begin{enumerate}
    \item for each $ i \in \range k$ it holds that $ 0 \leq \sum_{j=1}^i \iota_j \leq |q|$;
    \item $\sum_{j=1}^k \iota_j = |q|$;
    \item for each $ i \in \range k$ it holds that $A_i = 
    \begin{cases}
        Q_{s_{i}+1}\;\;\; \text{ if } \iota_i = 1 \\
        Q_{s_{i}}\;\;\;\;\;\; \text { ~~if } \iota_i = -1
    \end{cases}$
    
    where $s_i = \sum_{j=1}^{i-1} \iota_j$ and $Q_j$ is the $j$-th symbol in $q$.
\end{enumerate}
\end{definition}

\noindent
Our path in $G_{q,V}$, leading from $\varepsilon$ to $q$, induces, in a natural way, a $q$-walk\footnote{If $w \in \Sigma^*$, by $w^{-1}$ we mean $w$ reversed with every letter $\alpha$ replaced with $\alpha^{-1}$.}~~
$ (v_{p_1})^{\epsilon_1}(v_{p_2})^{\epsilon_2}\ldots (v_{p_m})^{\epsilon_m} $.
%
%
For clarity, let us illustrate this with:

\begin{example}
Imagine that $q = ABCD$ and $V = \set{ABC, BC, BCD}$. Then there is a path $ \varepsilon \to ABC \to A \to ABCD $ in $G_{q,V}$.
This path induces a $q$-walk   $(ABC)(BC)^{-1}(BCD)$, which is equal to $ABCC^{-1}B^{-1}BCD$.

\end{example}

Now we are going to explain how each $q$-walk 
can be turned into $q$ by a sequence of simple reductions:

\newcommand{\thue}{\;\xrightarrow
{\parbox{3mm}{\footnotesize \hspace{1.5mm}+/-\vspace{-5mm}}}
}

\newcommand{\thuetrans}{\xrightarrow
{\parbox{3mm}{\footnotesize \hspace{1.5mm}+/-\vspace{-5mm}}}
{\parbox{1mm}{\footnotesize \hspace{-5mm}*\vspace{2mm}}\hspace{-1mm}}
}

\newcommand{\theu}{\;\xrightarrow
{\parbox{3mm}{\footnotesize \hspace{1.5mm}-/+\vspace{-5mm}}}
}

\newcommand{\theutrans}{\xrightarrow
{\parbox{3mm}{\footnotesize \hspace{1.5mm}-/+\vspace{-5mm}}}
{\parbox{1mm}{\footnotesize \hspace{-5mm}*\vspace{2mm}}}\hspace{-1mm}
}

\begin{definition}
For any $w,w' \in \bar\schema^*$ and for any $A\in \schema$ we define:

\begin{center}
$ wAA^{-1}w' 
\thue
ww' \;\;\;\;\;$
and
$\;\;\;\;\; wA^{-1}Aw' 
\theu
ww' $.
\end{center}\vspace{2mm}

Relations $\thuetrans$ and $\theutrans$ are defined as the reflexive transitive closure of
 $\thue$ and of $\theu$, respectively.\\
\end{definition}

%

\begin{lemma}\label{niby-thue}
If $w\in\bar\schema^*$ is a $q$-walk, then
$ w
\thuetrans
q \;\;$
and
$\;\;w 
\theutrans
q $.
\end{lemma}

For the proof of Lemma \ref{niby-thue} see Appendix C.


\subsection{Seeing $D$ (and $D'$) as relations in ${\mathbb Q}^n\times {\mathbb Q}^n$.}

\begin{definition}
Let $R \in \schema$. Then $M_R^D$ is the incidence matrix of the relation $R$ in structure $D$, 
that is $M_R^D \in \Q^{n \times n}$ and
$ M_R^D(i, j) = 1 $ if and only if  $R(a_i, a_j) \in D $ and $ M_R^D(i, j) = 0 $ if and only if 
$R(a_i, a_j) \notin D$.
\end{definition}
\begin{definition}
Let $w \in \schema^*$. Then we define a matrix $M_w^D \in {\mathbb Q}^{n \times n}$ in the  inductive way:  \hspace{4mm}
 {\large \textbullet~~} $M_\varepsilon^D$ is the $n\times n$ identity matrix.\\
 \hspace*{31mm} {\large \textbullet~~} For $R \in \schema, w \in \schema^*$, $M^D_{Rw} = M^D_R M^D_w$.

\end{definition}

It is well-known that:

\begin{fact}\label{superexpress}
If $w \in \schema^*$ then $w(D)[a_i, a_j] = M_w^D(i, j)$. 
\end{fact}

Matrices $M^{D'}_R$ and $M^{D'}_w$ are defined analogously, and, obviously, Fact \ref{superexpress} remains
true for them. 

Of course in general we cannot assume that $M^{D}_w=M^{D'}_w$ for $w\in \Sigma^*$. 
But, since for each $v\in V$ we have $v(D)=v(D')$, we know that  for each $v\in V$ we have
 $M_v^D = M_v^{D'}$, so we can write just $M_v$ instead of $M_v^D$ or $M_v^{D'}$. 
Recall that  we need to show that $M_q^D = M_q^{D'}$. So, if we manage to somehow present 
 $M_q^D$ (and hence also $M_q^{D'}$) as a function of arguments $\{ M_v \mid v \in V \}$, then we are done.

Let us also remark that if $M_v$ were invertible, for all $v \in V$,  then it would be  easy 
to see that $M_q^D = M_{v_{p_1}}^{\epsilon_1} ... M_{v_{p_m}}^{\epsilon_m}$ and likewise $M_q^{D'}$. 
However, in the general case, there is of course no 
reason to think that the matrices $M_v$ are invertible, and thus we  need our argument to be
 a little bit  more sophisticated.

Now the matrices will be understood as linear functions. And these functions will be understood as relations. 
And, while  we know that not all matrices are invertible, and in consequence not all the functions under consideration are, 
relations  can always be inverted!

By $I$ we will denote the identity relation: $I = \{ \pair{x, x} \mid x \in {\mathbb Q}^n \}$.

\begin{definition}
\begin{enumerate}
    \item For a matrix $M \in {\mathbb Q}^{n \times n}$ let the function $h_M : {\mathbb Q}^n \to {\mathbb Q}^n$ be defined as $h_M(v) = Mv$. 
    \item For a function $f$ let $\rel f$ denote the relation equal\footnote{This means that $\rel f = \{ \pair{x, y} \mid f(x) = y \}$. We make such  distinction 
    since composition and inversion work for functions slightly differently than for relations.} to $f$.
    \item For $R \in \schema$ let $H_R = \rel{h_{M_R^D}}$ and $ H_{R^{-1}} = H_R^{-1}$
    \item For $w \in \bar\schema^*$ we define $H_w$ inductively:\\
    {\large\textbullet}\hspace{2mm} $H_\varepsilon = I$ \hspace{15mm}
   {\large\textbullet}\hspace{2mm} $H_{\alpha w} = H_w H_\alpha$ for $\alpha \in \bar\schema, w \in \bar\schema^*$
    
\end{enumerate}
\end{definition}

The relations $H_w$ depend on $D$ (in the sense that they would not be equal if we computed them in $D'$ instead of $D$), so 
 the reader may think that there should be $H^D_w$ instead of $H_w$. But omitting the superscript leads to no confusion: 
$D$ is the only structure for which the relations $H_w$ are ever considered.

\begin{observation}
For $w \in \schema^*$, $H_w = \rel{h_{M_w^D}}$ and $H_{w^{-1}} = (H_w)^{-1}$
\end{observation}

For the proof of the Observation use (easy) induction and the fact that for $w, w' \in \schema^*$ it holds that 
$H_{ww'} = H_{w'}H_w = \rel{h_{M_{w'}^D}} \rel{h_{M_w^D}} = \rel{h_{M_w^D} \circ h_{M_{w'}^D}} = \rel{h_{M_{ww'}^D}}$. 
\hfill $\square$\\

\noindent
It is well-known that the correspondence $M \mapsto h_M$ is 1-1. Also the correspondence $f \mapsto  \rel f$ is 1-1.
So in order to represent  $M_q^D$  as a function of arguments $\{ M_v \mid v \in V \}$ it is enough
to represent $H_q$ as a function of $\{H_v\mid  v \in V \}$. Which we do in the next subsection.

\subsection{Using Lemma \ref{niby-thue}.}

Let us start this subsection with a really very simple lemma:

\begin{lemma}\label{zawieranie-id}
Let $f : {\mathbb Q}^n \to {\mathbb Q}^n$. Then
 $\rel f ~ (\rel f)^{-1} \supseteq I$
 and  $(\rel f)^{-1} ~ \rel f \subseteq I$.

\end{lemma}

\noindent
{\sc Proof.}
$\rel f ~ (\rel f)^{-1} = \{ \pair{x, y} \mid \exists z \; f(x) = z \land f(y) = z \} =$\\ $\;\;\;\;\;\;\;=\{ \pair{x, y} \mid f(x) = f(y)\} \supseteq I $ \vspace{2mm}

\noindent
$(\rel f)^{-1} ~ \rel f = \{ \pair{x, y} \mid \exists z \; f(z) = x \land f(z) = y \} =$\\  $\;\;\;\;\;\;\;=\{ \pair{x, x} \mid \exists z ~ f(z) = x \} \subseteq I $ \hfill $\square$\\

Now we will see what the relations $H_w$ are good for:

\begin{lemma}\label{cor:reductions}
For $u,u' \in \bar\Sigma^*$:
\begin{enumerate}
    \item if  $u\thue u'$ then $H_{u} \subseteq H_{u'}$;
    \item  if $u\theu u'$ then $H_{u} \supseteq H_{u'}$.
\end{enumerate}
\end{lemma}

For the proof of Lemma \ref{cor:reductions} see Appendix C. 
Notice that, for a $q$-walk $w$, Lemmas \ref{cor:reductions} and \ref{niby-thue} give us two approximations of $H_w$:

\begin{lemma}
If $w$ is a \emph{$q$-walk}, then $H_q\subseteq  H_w \subseteq H_q$. \hfill $\square$
\end{lemma}

\noindent
Our next corollary is certainly not going to come as a surprise:

\begin{corollary}
If $w$ is a \emph{$q$-walk}, then $H_q= H_w$. 
\end{corollary}

Now, recall that
$ (v_{p_1})^{\epsilon_1}(v_{p_2})^{\epsilon_2}\ldots (v_{p_m})^{\epsilon_m} $
is a $q$-walk. So, by the last corollary
$ H_q \; = \; H_{v_{p_m}}^{\epsilon_m}\ldots H_{v_{p_2}}^{\epsilon_2}H_{v_{p_1}}^{\epsilon_1} $.

Which shows that  $H_q$ is indeed a function of  $\{H_v\mid  v \in V \}$ and ends the proof of Lemma \ref{lem:path-queries}($\Leftarrow$) and of Theorem \ref{th-1}.


\section{The boolean case. Our main results.}\label{sec-main-th-1}

In contrast to the set-semantics world, where determinacy is easily decidable for unary UCQs, 
and trivially decidable for boolean UCQs, in the multiset  setting already the
boolean case is undecidable:

\begin{theorem} The problem whether, for a set 
 $V$ of boolean UCQs and for another boolean UCQ $\query$, it holds that $V\detbag \query$ is undecidable. 
  \end{theorem}

\noindent
This negative result is not really hard to prove (see Appendix A). 
The main technical result of this paper, however, is:

\begin{theorem}
The problem whether, for a given set $V_0$ of  boolean conjunctive queries and for a given  boolean conjunctive query $q$, 
it holds that $V_0 \detbag q$, 
is decidable. 
\end{theorem}

\noindent
{\bf \small The rest of this section, and Sections \ref{sec-main-th-3}-\ref{usinggoodS-proof} are devoted to the proof of \cref{theorem-3}. 
 A set $V_0$ of boolean conjunctive queries and a boolean conjunctive query $q$ are fixed from now on.}

\begin{definition}
Let $V$ be the set  $\{ v \in V_0 \mid q \subseteq_\sett v \}$.  Let us also denote $V' = V \cup \{ q \}$.
\end{definition} 

Queries from $V$ are the ones that cannot return 0 in any interesting (from the point of view of this proof)
structure $D$. Queries from $V_\setminus V$ are free to return 0, and they actually will.

\begin{observation}\label{najprostsza}
If $D$ is any structure, $v\in V$ and $v(D)= 0$ then also  $q(D)= 0$.
\end{observation}

\begin{definition}\label{def-wu}
Let $W = \{ w_1, ..., w_k \}$ be the 
set\footnote{When we say ``set'' we mean that each such connected component only occurs once in $W$. And we think that isomorphic structures are equal.} of all connected components of the query\footnote{$\sum$ is an operation on structures here, as defined in Section \ref{gtt}.} $\sum_{v \in V'} v$. In other words, $W$ is a 
minimal set of structures such that for every connected component $u$ of $\sum_{v \in V'} v$ there is $w \in W$ isomorphic to $u$. From now on, the letter
$k$ will always denote the cardinality of $W$.
\end{definition}

Queries from $W$ are going to serve us as \emph{basis}  queries, in the linear algebra sense:

\begin{observation}\label{reprezentacja-wek}
Let $v \in V'$. Then $v = \sum_{i=1}^k a_i w_i$ for some $a_1, ..., a_k \in \N$.
\end{observation}

Note, that such representation is unique. Thus:

\begin{definition}
 For a query $v\in V'$ we define the \emph{vector representation of $v$} as  $\vec v = \begin{bmatrix} a_1 \\ \vdots \\ a_k \end{bmatrix} $, where $a_1, ..., a_k$ are as in \cref{reprezentacja-wek}.
 \end{definition}
 
So all the queries of interest are now seen as vectors in some $k$-dimensional vector space.

\begin{observation}\label{obs:query-eval}
If $D$ is any structure and $ v \in V'$ then $v(D) = \prod_{i = 1}^k w_i(D)^{\vec v(i)}$.
\end{observation}

\noindent
{\em Proof:~}
Notice that
$v(D) = |hom(v, D)| = |hom(\sum_{i=1}^k \vec v(i) w_i, D)| =$ $\prod_{i = 1}^k w_i(D)^{\vec v(i)}$. 
The last equality follows from \cref{lem:hom-properties}.
\hfill $\square$

\vspace{1mm}
Now we are ready for our  {\bf Main Lemma}:

\begin{lemma}\label{main-lemma-3}
$V_0 \detbag q$ if and only if $\vv q \in \linspan \{ \vv v \mid v \in V \}$.
\end{lemma}

\noindent
Clearly, \cref{theorem-3} easily follows from \cref{main-lemma-3} 
as finding $V$ is of course decidable (in $\Sigma^P_2$ -- we first need to guess a set of homomorphisms and then check that we 
guessed all of them), while finding $W$ and testing whether  $\vv q \in \linspan \{ \vv v \mid v \in V \}$ are polynomial.

In Sections \ref{sec-main-th-3}--\ref{usinggoodS-proof} we present the (more complicated) $(\Rightarrow) $ part of the 
proof of  \cref{main-lemma-3}. The  (much easier)  $(\Leftarrow)$ part  is deferred to 
Appendix D. We however illustrate the idea of the $(\Leftarrow)$ part with:

\begin{example}
Let $w_1, w_2, w_3$ be some non-empty, pairwise non-isomorphic structures and let:  \hspace{15mm}
$ q = w_1 + w_2 + 2 w_3 $\\
   \hspace*{15mm}  $ v_1 = 2 w_1 + w_2 + 3 w_3 $ \hspace{10mm}
  $ v_2 = 5 w_1 + 2 w_2 + 7 w_3 $

  \vspace{1mm}
Then for a structure $D$: \hfill
$q(D) = w_1(D) w_2(D) w_3(D)^2 $,\\
$   v_1(D) = w_1(D)^2 w_2(D) w_3(D)^3 $
   ~and~ $\; v_2(D) = w_1(D)^5 w_2(D)^2 w_3(D)^7$.

    \vspace{1mm}
 If $v_2(D) \neq 0$, then $ q(D) = v_1(D)^3 / v_2(D) $
  so it is uniquely determined by $v_1(D), v_2(D)$. This equality corresponds to the equality of vector representations 
  $ \vec q = 3 \vec v_1 - \vec v_2$.
  If $v_2(D) = 0$, then for some $i \in \{ 1, 2, 3 \}, w_i(D) = 0$, so $q(D) = 0$ and it is determined again.

\end{example}

It easily follows from Lemma \ref{main-lemma-3} that in the very specific case of connected queries 
no non-trivial determinacy is possible:

\begin{corollary}\label{wnioseczek}
If all the queries in $V_0$ are connected, and $q$ is connected, then
$V_0 \detbag q$ if and only if $q \in  V_0$.
\end{corollary}

\section{Proof of Lemma \ref{main-lemma-3} $(\Rightarrow) $. Part 1.}\label{sec-main-th-3}

In this section we assume that $\vec q \notin \linspan \{ \vec v \mid v \in V \}$. And we are going to
show that $V_0 \not\detbag q$. To this end we need to find a pair of structures $D$ and $D'$ which is a counterexample for determinacy,
which means that:
\begin{enumerate}
    \item[(A)] $q(D) \neq q(D')$
    \item[(B)] $\forall v \in V \;\; v(D) = v(D')$ \hspace{5mm} 
    (B0) $\forall v \in V_0 \setminus V \;\; v(D) = v(D')$
\end{enumerate}

Notice that there is nothing in Definition \ref{pierwsza} that would tell us where to look for such a counterexample: $D$ and $D'$ are just {\bf any} structures 
in this definition. Our main discovery is that if such  $D$ and $D'$, forming a counterexample, can be found at all, then a
counterexample
 can also be found  in some  $k$-dimensional\footnote{Recall, that $k$ denotes, as always, the cardinality of $W$, the set of basis queries.} vector space that we are now going to introduce. And this is convenient, because  living in  a vector space one can use linear algebra tools.

\begin{definition}
For any set $S$ of $k$ structures (call them \emph{basis structures})
 let $\Ss$ be the set of all structures which can be  represented as sums of elements of $S$, 
that is $\Ss = \linspan^\N(S)$.
\end{definition}

Now, the totally informal idea is as follows. We know that $\vec q \notin \linspan \{ \vec v \mid v \in V \}$.
So there is vector $\vec z$ which is orthogonal to  $\vec v$ for each $v\in V$ but not to $\vec q$. Let us 
{\em somehow} define $D$ and $D'$ in such a way that  $\vec z$ is {\em ``the difference''} between $D$ and $D'$.
Then none of the  $v\in V$ will spot the difference between  $D$ and $D'$ but $q$ will.

\begin{definition}
Set $S$ of basis structures is \emph{decent} if for each $s \in S$ and for each $v \in V_0 \setminus V$ we have $ v(s) = 0$.
\end{definition}

It is easy to see that: 

\begin{observation}
If $S$ is \emph{decent}, then for each $D\in \Ss$ and for each  $v \in V_0 \setminus V$ we have $v(D)=0$. In consequence, 
if $S$ is \emph{decent}, then any pair $D, D'$ of structures from $\Ss$  satisfies condition (B0) above.
\end{observation}

\begin{definition}
\JCom{Uogólniłem trochę definicję, żeby można było używać tego pojęcia do $S^{(1)}$ w dowodzie \cref{goodS}}
For a set of structures $S = \{s_1, ..., s_{|S|}\}$ we define its \emph{evaluation matrix} $M_S \in \R^{k \times |S|}$ by the formula $M_S(i, j) = w_i(s_j)$.
\end{definition}

In other words, the $(i,j)$-entry of  $M_S$ is defined as the number of homomorphisms from $w_i$ to $s_j$.

\begin{definition}
$S$ is {\em good} when
 $S$ is decent and
 $M_S$ is nonsingular.

\end{definition}

Recall that the  set $W$, consisting of $k$ queries, is also a set of $k$ structures. What would happen if we just took $W$ as our set $S$ of basis structures? $W$ is of course always decent: if there were $w \in W, v \in V_0 \setminus V$ such that $v(w) > 0$, then, because $w(q) > 0$, we would get $v(q) > 0$. But $M_W$ is not always nonsingular:
\begin{SCfigure}[80]
    \includegraphics[scale=0.55]{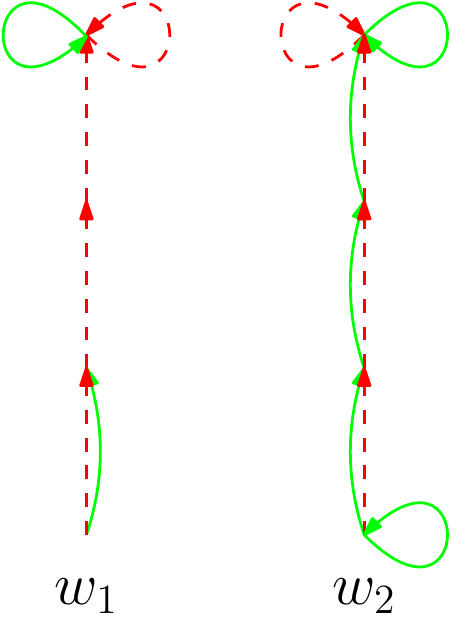}
    \caption{(\cref{ex:singular_M_example}) Here, $M_W$ is singular ($w_1$ and $w_2$ are structures over a schema 
    consisting of two binary relations, with $w_2$ having three additional green edges, compared to $w_1$).}
    \label{fig:singular_M_example}
\end{SCfigure}

    \begin{example}\label{ex:singular_M_example}
    Let $W$ consist of $w_1$ and $w_2$ as in~\cref{fig:singular_M_example}, then:
    $$M_W = \; \kbordermatrix{ & w_1 & w_2 \cr
          w_1 & 2 & 4 \cr
          w_2 & 1 & 2 }$$
          
    where $M_W(i,j) = \nhoms{w_i, w_j}$.
    \end{example}

Now, {\bf \small proof of \cref{main-lemma-3} will be completed once we  prove the following two lemmas}:

\begin{lemma}\label{goodS}
There exists a good set $S$ of basis structures.
\end{lemma}

\begin{lemma}\label{usinggoodS}
If $\vec q \notin \linspan \{ \vec v \mid v \in V \}$ and $S$ is a good set of basis structures then there exist 
$D,D'\in \Ss$  satisfying conditions $(A)$ and $(B)$ above. 
\end{lemma}

{\bf \small For the proof of Lemma \ref{goodS} see Section \ref{good-s-proof}, and for the proof of Lemma \ref{usinggoodS} see 
Section \ref{usinggoodS-proof}}.
Notice that Lemma \ref{usinggoodS} would not be true without the assumption that $M_S$ is nonsingular:

  \begin{example}\label{ex:W_is_bad_S}

    Let $q = w_1$ and $V_0 = \set{w_2}$ ($w_1, w_2$ are still as in \cref{fig:singular_M_example}), so that 
    (according to Definition \ref{def-wu})  $W = \set{w_1, w_2}$. Also, since $w_1 \subseteq_{set} w_2$,
    we get that $V=V_0$.
    
    Since  $\vec q \notin \linspan \{ \vec v \mid v \in V \}$,  it follows from our Main Lemma that $ V\not\detbag q $. 
   So couldn't we just take $S=W$ (notice that, since $V=V_0$, such $S$ is trivially decent)  and  look for the 
    counterexample structures $D$ and $D'$ in $\Ss=\linspan^{\mathbb N}{(S)}$?
    
    This would be in vain. For any structure $\db \in \Ss$ the equality $\nhoms{w_1 = q, \db} = 2\cdot \nhoms{w_2, \db}$ holds.
    So for any pair of structures 
  $\db, \db'\in \Ss$ there will be
    $\;\;V_0(\db) = V_0(\db') \Rightarrow q(\db) = q(\db').$
    \end{example}
    
   Let us however reiterate: the above example does not contradict our Main Lemma. It only shows that $S = \set{w_1, w_2}$ is 
   not {\bf \small good} enough to serve as a basis for a counterexample pair $D,D'$. 

\section{Proof of Lemma \ref{goodS}}\label{good-s-proof}

Our proof relies on the following lemma  stated in \cite{CV93} and  proved, as \cite{AAKPW21} report, in the paper \cite{F95}, which is not easy to access.

\begin{lemma}\label{co-lovasz}
Two structures $G, G'$  are isomorphic if and only if $ |hom(G, H)| = |hom(G', H)|$ for every structure\footnote{We of course assume here that all the structures in question are over some fixed relational schema.} $H$.
\end{lemma}

However, the proof of Lemma \ref{co-lovasz} is analogous to the proof of:

\begin{lemma}
Two structures $G, G'$  are isomorphic if and only if $ |hom(H, G)| = |hom(H, G')|$ for every structure $H$.
\end{lemma}
in the paper \cite{lovasz}, so we think we can skip it here.\\\vspace{-1.5mm}

A {\bf good $S$ will now be constructed} in four steps:\\\vspace{-1mm}

\noindent
\textbf{Step 1.} $S^{(1)} = \{ s_1^{(1)}, ..., s_m^{(1)} \}$ can be any finite set such that
$$ \forall w \neq w'\in W ~ \exists i \leq m \;\;\;  |hom(w, s_i^{(1)})| \neq |hom(w',s_i^{(1)})| $$
Such $S^{(1)}$ can be found thanks to Lemma \ref{co-lovasz}. Indeed, because elements of $W$ are pairwise non-isomorphic, for any $w \neq w'$ there is a structure $H$ such that $w(H) \neq w'(H)$ - it is enough to take one such structure for every pair $w \neq w' \in W$.

In the Steps 2-4 we  construct a good $S$ from $S^{(1)}$ using addition and multiplication of structures. And, by Lemma \ref{lem:hom-properties}, addition and multiplication of structures correspond to addition and multiplication (elementwise) of columns of the evaluation matrix. So this part is more about linear algebra  than about homomorphism counting.\\

\noindent
\textbf{Step 2.} Let $T\in \mathbb N$ be  greater than any element of the  matrix $M_{S^{(1)}}$. Then the set  $S^{(2)}$ consists of a single 
structure $ s^{(2)} $ where:\vspace{-1mm}
$$ s^{(2)} = \sum_{i=1}^m T^i s_i^{(1)} $$
~
\vspace{-4mm}
\begin{observation}\label{obs-krok-1}
\sloppy Suppose $w,w' \in W$ and $w \neq w'$. Then $|hom(w,s^{(2)})|  \neq |hom(w',s^{(2)})|$.
\end{observation}

For the proof of Observation \ref{obs-krok-1} see Appendix D.

\noindent
\textbf{Step 3.} Let now $S^{(3)} = \{ s_1^{(3)}, ..., s_k^{(3)} \}$ be a set of $k$ structures, 
where $s_i^{(3)} = \left(s^{(2)} \right)^{i-1}$. We are going to prove that the  matrix $M_{S^{(3)}}$ is nonsingular. 
Recall that $M_{S^{(3)}}(i, j) = |hom(w_i,s_j^{(3)})|$. 
Notice that: $$ |hom(w_i,s_j^{(3)})| = |hom(w_i,\left( s^{(2)} \right)^{j-1})| = |hom(w_i,s^{(2)})|^{j-1}$$

Then use the following lemma, which is proven in Appendix D:

\begin{lemma}\label{ilepierwiastkow}
Let $a_1, ..., a_k$ be pairwise-distinct real numbers. Then the matrix $A \in \R^{k \times k}$ defined as $ A(i, j) = a_i^{j-1}$ is nonsingular.
\end{lemma}

\noindent
\textbf{Step 4.} Now, $S^{(3)}$ is almost good. Almost, because we are still not sure if it is \emph{decent}. 
So let $S^{(4)} = \{ s_1^{(4)}, ..., s_k^{(4)} \}$ where $s_i^{(4)} = s_i^{(3)} \times q$. Observe that $M_{S^{(4)}}$ is just $M_{S^{(3)}}$ where $i$-th row has been multiplied by $w_i(q)$. However we know that $w_i$ is a subquery of some query $v \in V \cup \{ q \}$ - such a $v$ satisfies $v(q) > 0$. Therefore, $w_i(q) > 0$ and it is well-known that multiplying a matrix row by a non-zero factor doesn't affect its (non)singularity.

Let's observe that $S^{(4)}$ is \emph{decent}, that is, $\forall v \in V_0 \setminus V, s \in S^{(4)} ~ v(s) = 0$. Indeed, any $s \in S$ is of form $s = s' \times q$, so for any $v \in V_0 \setminus V$ we have $v(s) = v(s')v(q)$ and, by definition of $V$, $v(q) = 0$.


To sum up, we have found a good set of basis structures $S^{(4)}$. 
{\bf From now on,} we put $S = S^{(4)}$ and $ s_i = s_i^{(4)}$ for $i \in \range{k}$. We will also write $M$ instead of $M_S$.

\section{Proof of Lemma \ref{usinggoodS}}\label{usinggoodS-proof}

First some formulae which we will need:

\begin{definition}[Vector representation of a structure $s \in \Ss$]
\hspace{1mm}
For $s = \sum_{i=1}^k a_i s_i$ we define $\vec s =  \begin{bmatrix} a_1 \\ \vdots \\ a_k \end{bmatrix}$
\end{definition}

\begin{definition}
\hspace{1mm}
\begin{enumerate}
    \setlength\itemsep{0.8em}
    \item Let $\vec u, \vec v \in \R^k$. Then
    $\vec u \circ \vec v = \begin{bmatrix} \vec u(1)\vec v(1) \\ \vdots \\ \vec u(k)\vec v(k) \end{bmatrix}$
    
    \item Let $\vec u, \vec v \in \R_{\geq 0}^k$. Then 
    $\vecpow{\vec u}{\vec v} = \prod_{i=1}^k \vec u(i)^{\vec v(i)}$. 
    
    \item Let $t \in \R_+, \vec u \in \R^k$. Then
    $t^{\vec u} = \begin{bmatrix} t^{\vec u(1)} \\ \vdots \\ t^{\vec u(k)} \end{bmatrix}$
\end{enumerate}
\end{definition}
\begin{observation}\label{obs:vecpow}
\begin{enumerate}
    \item $ \vecpow{(\vec u \circ \vec v)}{\vec w} = (\vecpow{\vec u}{\vec w})(\vecpow{\vec v}{\vec w}) $ 
    \item $ \vecpow{t^{\vec u}}{\vec v} = t^{\dotprod{\vec u}{\vec v}} $
\end{enumerate}
\end{observation}

\begin{lemma}\label{lem:eval}
Let $v \in V, s \in \Ss$. Then $v(s) = \vecpow{(M \vec s)}{\vec v}$.
\end{lemma}

For the proof of \cref{lem:eval} see Appendix D.

\subsection{The set $\Ps$ and the cone $\Cs$}

So far the objects of our interest in this proof lived in two $k$-dimensional vector spaces. One was 
the space of queries, with $W$ as the basis.  Another one was the space $\Ss$ of structures, with basis $S$, where 
we are looking for the candidates for $D$ and $D'$. 

But we also need the third such $k$-dimensional space. Imagine you take some structure $s\in \Ss$. And you ask
what will be the results of applying the $k$ queries from $W$ to $s$. What you get is a $k$-dimensional vector
of natural numbers, which lives in the space of all possible (and impossible) answer vectors.

\begin{definition}
 $\Ps = \{ M \vec s \mid s \in \Ss \} = \{ M \vec u \mid \vec u \in \N^k \}$
  \end{definition}
  
$\Ps$ is the subset of our new space consisting of the actual answer vectors, generated by real structures 
  from $\Ss$. A related notion is:

\begin{definition}
$\Cs = \linspan^{\R_{\geq 0}} \{ M \vec s \mid s \in S \} = \linspan^{\R_{\geq 0}}\{ M e_i \mid i \in \range{k} \}$. In other words, $\Cs$ is a convex cone generated by basis standard vectors multiplied by matrix $M$.
\end{definition}

The following easy observation shows that $\Ps$ is a subset of $\Cs$. A proper subset, since
only vectors of natural numbers can be in $\Ps$. And there is even no reason to think that $\Ps=\Cs\cap  \N^k$.

\begin{observation}[easy]
$ \Cs = M(\R_{\geq 0}^k) = \linspan^{\R_{\geq 0}} \{ M \vec s \mid s \in \Ss \} $
\end{observation}

\begin{example}\label{ex:C-and-P}
Let $w_1, w_2$ be as in \cref{fig:singular_M_example}. 
Let $s_1$  be a single vertex, with red and green loops and let $ s_2 = w_2$  Then:\vspace{-2mm}
$$M_S = \; \kbordermatrix{ & s_1 & s_2 \cr
          w_1 & 1 & 4 \cr
          w_2 & 1 & 2 }$$

Then $\Cs$ and $\Ps$ are as in $\cref{fig:C-and-P}$. Notice that $M_S$ is now non-singular. This observation 
is not unrelated to the fact that the grey area in $\cref{fig:C-and-P}$ has non-empty interior.
\end{example}

\begin{SCfigure}
    \centering
    \begin{tikzpicture}[scale=0.5]
      \filldraw[thin,gray!20] (0,0)--(6,3)--(6,6)--cycle;
      \draw[thin,gray!40] grid (6,6);
      \filldraw [red!40] (0,0) circle (2pt);
      \filldraw [red!40] (4,2) circle (2pt);
      \filldraw [red!40] (1,1) circle (2pt);
      \filldraw [red!40] (2,2) circle (2pt);
      \filldraw [red!40] (3,3) circle (2pt);
      \filldraw [red!40] (4,4) circle (2pt);
      \filldraw [red!40] (5,5) circle (2pt);
      \filldraw [red!40] (6,6) circle (2pt);
      \filldraw [red!40] (5,3) circle (2pt);
      \filldraw [red!40] (6,4) circle (2pt);
      \draw[->] (0,0)--(6,0) node[right]{$x$};
      \draw[->] (0,0)--(0,6) node[above]{$y$};
      \draw[line width=1pt,-stealth](0,0)--(4,2);
      \draw[line width=1pt,-stealth](0,0)--(1,1);
    \end{tikzpicture}
    \caption{(\cref{ex:C-and-P})  The $x$-coordinate represents the answer to $w_1$, and the $y$-coordinate to $w_2$. Red dots correspond to $\Ps$ and the gray area to $\Cs$. The arrows represent column vectors of $M_S$.}
    \label{fig:C-and-P}
\end{SCfigure}
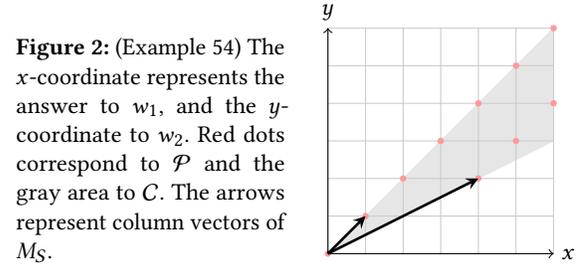

\subsection{It is here where things finally happen}

We spent several pages  pushing a rabbit into the hat. Now we are finally going to pull it out.
First, we notice that while not all the vectors in $\Cs \cap \Q^k$ are in $\Ps$, all of them are somehow related to $\Ps$:

\begin{lemma}\label{lem:mult-by-c}
Let $\vec u \in \Cs \cap \Q^k$. Then there exists $ c\in \N_+$ with $c \vec u \in \Ps$.
\end{lemma}
\vspace{-3mm}
\begin{proof}
$M$ is nonsingular\footnote{Finally we are using this nonsingularity. But it is 
 the proof of Lemma \ref{lem:exist-p} where it is really fundamentally needed.} so there exists $M^{-1}\in \Q^{k \times k}$.
 Let $\vec \alpha = M^{-1} \vec u$. Since  $\vec u \in \Cs \cap \Q^k$ we have that $\vec \alpha \in \Cs \cap \Q^k$.
Since $\vec u \in \{ M \vec v \mid \vec v \in \R_{\geq 0}^k \}$ we know that
$\vec \alpha \in \R_{\geq 0}^k $.
So there is $c \in \N_+$ such that\footnote{Take a common multiple of all denominators of 
the coordinates of
$\vec \alpha$.}  $c \vec \alpha \in \N_+$. Now, $c \vec u = c (M \vec \alpha) = M (c \vec \alpha) \in \Ps$.
\end{proof}

\begin{lemma}\label{lem:exist-p}
There are $\vec p, \vec {p'} \in \Cs \cap \Q^k$ such that:

\hspace*{6mm}(1) $\forall v \in V ~ \vecpow{\vec p}{\vec v} = \vecpow{\vec {p'}}{\vec v}$ \hfill
  (2)  $\vecpow{\vec p}{\vec q} \neq \vecpow{\vec {p'}}{\vec q}$ \hspace*{6mm}

\end{lemma}

\noindent
Before we prove {\bf \small Lemma \ref{lem:exist-p}} let us show that it {\bf \small implies Lemma \ref{usinggoodS}}:\\
Indeed, if we find $p, p'$ as in \cref{lem:exist-p}, then, by \cref{lem:mult-by-c} we can find $c, c' \in \N_+$ such that $c\vec{p}, c'\vec{p'} \in \Ps$. Of course then $cc'\vec{p}, cc'\vec{p'} \in \Ps$ too. Let $\vec{cc'} = \begin{bmatrix} cc' \\ \vdots \\ cc' \end{bmatrix}$. Then, for $v \in V \cup \{ q \}$ we have 

\begin{align*} 
&(\vecpow{cc' \vec p}{\vec v}) - (\vecpow{cc' \vec {p'}}{\vec v}) = \\
&= (\vecpow{\vec{cc'}}{\vec v})(\vecpow{\vec p}{\vec v}) - (\vecpow{\vec{cc'}}{\vec v})(\vecpow{\vec {p'}}{\vec v}) & \text{(by \cref{obs:vecpow})} \\
&= (\vecpow{\vec{cc'}}{\vec v})((\vecpow{\vec p}{\vec v}) - (\vecpow{\vec {p'}}{\vec v}))
\end{align*}
Because $cc' > 0$, also $(\vecpow{\vec{cc'}}{\vec v}) > 0$ and we get:\\
$ (\vecpow{cc' \vec p}{\vec v}) - (\vecpow{cc' \vec {p'}}{\vec v}) = 0\;\;\;\;$ iff $\;\;\;\;(\vecpow{\vec p}{\vec v}) - (\vecpow{\vec {p'}}{\vec v}) = 0$.\\
Then we take $s, s' \in \Ss$ such that $M \vec s = cc' \vec p, M \vec {s'} = cc' \vec {p'}$. By \cref{lem:eval} we have:
\begin{enumerate}
    \item for $v \in V$, $ v(s) = \vecpow{cc' \vec p}{\vec v} = \vecpow{cc' \vec {p'}}{\vec v} = v(s') $
    \item $ q(s) = \vecpow{cc' \vec p}{\vec q} \neq \vecpow{cc' \vec {p'}}{\vec q} = q(s') $
\end{enumerate}
So $D = s, D' = s'$ are the structures as postulated by Lemma \ref{usinggoodS}.\\

\vspace{-3mm}
\noindent
The last thing needed for the proof of Lemma \ref{usinggoodS} is:\\
\vspace{-3mm}

\noindent
{\sc proof of Lemma \ref{lem:exist-p}.}
Take $\vec {z_0} \in \Q^k$ such that (1) $\forall v \in V ~ \dotprod{\vec {z_0}}{\vec v} = 0$ and (2) $\dotprod{\vec {z_0}}{\vec q} \neq 0$.
Such $\vec {z_0}$ exists thanks to  \cref{fact:orthogonal-vector} and to 
the assumption that $\vec q \notin \linspan \{ \vec v \mid v \in V \}$. Then take $d \in \N_+$ such that $d \vec {z_0} \in \Z^k$. Let $\vec z = d \vec {z_0}$. Clearly, $\vec z$ satisfies conditions (1) and (2) too.

Let $\vec p$ and $ r$ be as in \cref{cor:interior}. This means that   $\vec p \in \Cs \cap \Q^k$  and the ball with center $\vec p$ and radius $r$ is contained in $\Cs$, with $r>0$. So we already have $\vec p$ and we will find $\vec{p'}$ in this ball.
\vspace{-1mm}
\begin{lemma}\label{lemat-ce}
There exists $t \in \R_+ \setminus \{ 1 \}$ such that $t^{\vec z} \circ \vec p \in \Cs \cap \Q^k$.
\end{lemma}
\vspace{-3mm}
\begin{proof}
Observe that the function $t \mapsto t^{\vec z} \circ \vec p$ is continuous and it maps $1$ to $\vec p$. 
Thus there is $\delta > 0$ such that:
$$ \forall t \in (1 - \delta, 1 + \delta) \;\;  \| t^{\vec z} \circ \vec p - \vec p \| < r $$
It is now enough\footnote{
It is here where one needs $\vec z$ to be in  $\Z^k$, otherwise $t^{\vec z}$ would not be rational.} to take any $t \neq 1$ in $(1 - \delta, 1 + \delta) \cap \Q$.
\end{proof}

Let $\vec{p'} = t^{\vec z} \circ p$ where $t$ is as in Lemma \ref{lemat-ce}. By \cref{obs:vecpow}:

\noindent
{\large\textbullet~}For $v \in V$,
    $\vecpow{\vec{p'}}{\vec v} = \vecpow{(t^{\vec z} \circ p)}{\vec v} = (\vecpow{t^{\vec z}}{\vec v})(\vecpow{\vec p}{\vec v}) = t^{\dotprod{\vec z}{\vec v}} (\vecpow{\vec p}{\vec v}) = \vecpow{\vec p}{\vec v}$\\
   {\textbullet~}~$\vecpow{\vec{p'}}{\vec q} = \vecpow{(t^{\vec z} \circ p)}{\vec q} = (\vecpow{t^{\vec z}}{\vec q})(\vecpow{\vec p}{\vec q}) = t^{\dotprod{\vec z}{\vec q}} (\vecpow{\vec p}{\vec q}) \neq \vecpow{\vec p}{\vec q}$

   \noindent
This ends the proof of Lemma \ref{lem:exist-p}, of Lemma \ref{usinggoodS} and of Theorem \ref{th-3}.



\bibliography{references}
\bibliographystyle{ieeetr}

\newpage

\input{UCQs}

\end{document}

%% file: packages.tex
\usepackage{color}
\usepackage{algorithmicx}
\usepackage{amsmath}
\usepackage{amsthm}
\usepackage{kbordermatrix}
\usepackage{bbm}
\usepackage{pifont}
\usepackage{wasysym}
\usepackage{enumitem}
\usepackage{marvosym}
\usepackage{tikz}
\usepackage{sidecap}
\usepackage[labelfont=bf, textfont=up]{caption}

\usepackage[normalem]{ulem}
\usepackage[capitalise]{cleveref}

%% file: macros.tex
\newtheorem{lemma}{Lemma}
\newtheorem{definition}[lemma]{Definition}
\newtheorem{example}[lemma]{Example}

\newtheorem{theorem}{Theorem}

\newtheorem{problem}[lemma]{Problem}

\newtheorem{corollary}[lemma]{Corollary}
\newtheorem{observation}[lemma]{Observation}

\newtheorem{fact}[lemma]{Fact}

\newcommand{\overbar}[1]{\mkern 1.5mu\overline{\mkern-1.5mu#1\mkern-1.5mu}\mkern 1.5mu}

\newcommand{\sett}{\text{set}}

\newcommand{\detbag}{\xrightarrow{\parbox{5mm}{\tiny \hspace{2mm}bag\vspace{-1mm}}}}

\newcommand{\detset}{\xrightarrow{\parbox{5mm}{\tiny \hspace{2mm}set\vspace{-1mm}}}}

\newcommand{\detttt}{\xrightarrow[]{}}
\newcommand{\range}[1]{\{ 1, \dots, #1 \}}

\newcommand{\N}{\mathbb{N}}
\newcommand{\R}{\mathbb{R}}
\newcommand{\Q}{\mathbb{Q}}
\newcommand{\Z}{\mathbb{Z}}
\newcommand{\Ss}{\mathcal{S}}
\newcommand{\Ps}{\mathcal{P}}
\newcommand{\Cs}{\mathcal{C}}

\newcommand{\vecpow}[2]{#1 \;\mars\; #2}

\newcommand{\linspan}{\text{span}}
\newcommand{\dotprod}[2]{\langle #1, #2 \rangle}

\newcommand{\set}[1]{\{#1\}}
\newcommand{\tuple}[1]{\langle#1\rangle}
\newcommand{\pair}[1]{\tuple{#1}}

\newcommand{\schema}{\Sigma}
\newcommand{\structure}{D}

\newcommand{\db}{\structure}

\newcommand{\views}{V}
\newcommand{\query}{q}

\newcommand{\domain}[1]{dom(#1)}

\newcommand{\homs}[1]{hom(#1)}
\newcommand{\nhoms}[1]{|\homs{#1}|}

\newcommand{\coeff}{c}
\newcommand{\hinstance}{I}
\newcommand{\hsolution}{f}

\newcommand{\Vector}[1]{\vec{#1}}
\newcommand{\Vx}{\Vector{x}}
\newcommand{\Vy}{\Vector{y}}

\newcommand{\Vt}{\Vector{t}}

\newcommand{\Va}{\Vector{a}}

\newcommand{\rel}[1]{\overbar{#1}}

\newcommand{\Pcom}[1]{{\color{orange}#1}}

\newcommand{\Pst}[1]{\sout{#1}}

\newcommand{\JMcom}[1]{{\color{green}#1}}

\newcommand{\JCom}[1]{{\color{red}#1}}

\newcommand{\Jst}[1]{\sout{#1}}

\renewcommand{\Pcom}[1]{}

\renewcommand{\Pst}[1]{}

\renewcommand{\JMcom}[1]{}

\renewcommand{\JCom}[1]{}

\renewcommand{\Jst}[1]{}

\newcommand{\querylang}{\mathcal{L}}

%% file: UCQs.tex

\section{Appendix A --- The UCQ case}\label{sec-UCQ}

This section is entirely devoted to the proof of \cref{th-2}. 

As our source of undecidability we take {\bf Hilbert's Tenth Problem}.
 It is well known that the following problem is undecidable:
\begin{problem}\label{problem:H10}
Given a polynomial equation with finite number of unknowns and integer coefficients, determine whether it has a solution such that every unknown is a natural number.
\end{problem}

An instances of Hilbert's Tenth Problem can be seen as a set of monomials (with integer  coefficients).
For a given monomial $m$  we will denote with $\coeff(m)$ its coefficient and with $m(x)$ 
we denote degree of $m$ with respect to  $x$ in $m$ (if $x$ in not present in $m$ then $m(x) = 0$).

In order to prove  \cref{th-2} we will construct a reduction from (the complement) of 
Hilbert's Tenth Problem. As an instance of that problem we are given a set $\hinstance = \set{m_1, m_2, \ldots, m_k}$ of monomials. Let $x_1, x_2, \ldots, x_n$ be the unknowns present in $\hinstance$. 
 We are going to produce  a schema $\schema$, a boolean UCQ $\query$ and a set  $\views$ of boolean UCQ such that $\hinstance$ has {\bf no} solution if and only 
 if $V\detbag q$.

We start with $\schema$, which will consist\footnote{One needs to mention here  that our nullary predicates  come from \cite{SV05} and \cite{M20} 
and our unary predicates come from \cite{IR95}.} of nullary and unary predicates: $H, C, X_1(x), \ldots, X_n(x)$.
For a structure $\db$ and for $R \in \schema$ let us denote, with $\db_{R}$,  the number of atoms of relation $R$ in structure $\db$.
Notice that, since $H$ and $C$ are nullary,   $\db_{H},\db_{C}\in \{0,1\}$  for each $\db$.

Now. the general idea, that one could have in mind, is that  the upcoming set of boolean CQs $\views$ will make sure that any pair of 
\emph{distinct} structures $\db, \db'$ over $\schema$ satisfying $\views(\db) = \views(\db')$ is equal on $X_i$ and differs on $H$ and $C$.

Before we can define $\views$ let us construct two UCQs $\Psi_N$ and $\Psi_P$. 
First, for every monomial $m$ we define the following boolean CQ:

$$\Phi_m = \exists^* \bigwedge_{X_i \in \Sigma}\bigwedge_{j = 1}^{m(x_i)} X_i(y_{i,j})$$

where the quantifier $\exists^*$ binds all the variables $y_{i,j}$ that occur in the formula.

For a structure $\db$  over $\schema$ and for $m\in \hinstance$ let $m_D$  be the value of 
$m$ after substituting, for each unknown $x_i$ in $m$,  the number $\db_{X_i}$. For a solution $\hsolution$ of instance $\hinstance$ we write $m_\hsolution$ to denote the value of $m$ after substituting each unknown $x_i$ with its value in solution $\hsolution$. \JCom{Trochę mnie boli, że piszemy tutaj strasznie długie zdania zamiast używać normalnej matematycznej notacji dla wielomianów: $m(D_{X_1}, D_{X_2}, ..., D_{X_n}), m(\bar f)$. Pewnie dlatego, żeby nie myliło się z $m(x)$, ale może można by to jeszcze zmienić}

\begin{lemma}\label{lem:ucq-phi-homs}
For each $D$ and each  $m\in \hinstance$:
$$ m_D = \coeff(m) \cdot \Phi_m(\db)$$
\end{lemma}
\begin{proof}
It follows from~\cref{lem:hom-properties}~(\ref{eq:hom-properties-5}).
\end{proof}

Let $P$ be subset of $\hinstance$ containing monomials with positive coefficients and let $N$ contain monomials with negative coefficients, then define: 
$$\Psi_P = \bigvee_{m \in P} \bigvee_{i = 1}^{\coeff(m)} \Phi_{m} \wedge H, \hspace{0.5cm} \Psi_N = \bigvee_{m \in N} \bigvee_{i = 1}^{\coeff(m)} \Phi_{m} \wedge C.$$

\begin{lemma}\label{lem:ucq-psi-p}
For each $D$ it holds that: 
$$\db_{H}\cdot \sum_{m \in P} m_D = \Psi_P(\db).$$
\end{lemma}
\begin{proof}
\begin{align*}
\Psi_P(\db) & = \sum_{m\in P}\sum_{i = 1}^{\coeff(m)} (\Phi_P \wedge H)(\db) \tag{\cref{lem:hom-properties}}\\
                    & = \sum_{m\in P}\sum_{i = 1}^{\coeff(m)} \db_H \cdot \Phi_P(\db)  \\
                    & = \db_H \cdot \sum_{m\in P}\sum_{i = 1}^{\coeff(m)} \Phi_P(\db) \\
                    & = \db_H \cdot \sum_{m\in P}\coeff(m)\cdot \Phi_P(\db) \\
                    & = \db_H \cdot \sum_{m\in P}m_\db \tag{\cref{lem:ucq-phi-homs}}
\end{align*}
\end{proof}

\begin{lemma}\label{lem:ucq-psi-n}
For each $D$ it holds that:
$$\db_{C}\cdot \sum_{m \in N} m_D = -\Psi_N(\db)$$
\end{lemma}
\begin{proof}
Analogous to the proof of~\cref{lem:ucq-psi-p}.
\end{proof}


Finally we are able to define a query $\query$ and a set of queries $\views$. Our boolean UCQ $\query$ will simply be equal to $H$. The set $\views$ will contain the following boolean UCQs :
\begin{itemize}
    \item $V_1 = H \vee C$,
    \item $V_{x_i} = \exists{y}\;X_i(y)$ for each $X_i$ in schema $\schema$,
    \item $V_\hinstance = \Psi_P \vee \Psi_N$.
\end{itemize}

The above definition of $\views$ implies the following property:
\begin{lemma}\label{lem:ucq-counter-property}
For every pair of distinct structures $\db, \db'$ such that $\views(\db) = \views(\db')$ the following holds:
    $$\db_{X_i} = \db'_{X_i}, \hspace{0.5cm}
    \db_{H} = \db'_{C},  \hspace{0.5cm}
    \db_{C} = \db'_{H}, \hspace{0.5cm}
    \db_{H} \neq \db_{C}$$

\end{lemma}
\begin{proof}
Property $\db_{X_i} = \db'_{X_i}$ is obvious given views $V_{x_i}$. From $V_1(\db) = V_1(\db')$ we get following possibilities:
\begin{enumerate}
    \item  $\quad \db_{H} = \db'_{C},  \quad \db_{C} = \db'_{H}, \quad \db_{H} \neq \db_{C}$ \hfill ($V_1(\db) = 1$) 
    \item  $\quad \db_{H} = \db'_{H},  \quad \db_{C} = \db'_{C}, \quad \db_{H} \neq \db_{C}$ \hfill ($V_1(\db) = 1$) 
    \item  $\quad \db_{H} = \db'_{H} = \db_{C} = \db'_{C} = 0$ \hfill ($V_1(\db) = 0$) 
    \item  $\quad \db_{H} = \db'_{H} = \db_{C} = \db'_{C} = 1$ \hfill ($V_1(\db) = 2$) 
\end{enumerate}

From $\db_{X_i} = \db'_{X_i}$ and the fact that $\db \neq \db'$ we conclude that only $(1)$ can hold.
\end{proof}

Thus whenever we will have two different structures $\db, \db'$ satisfying $\views(\db) = \views(\db')$ we will assume without loss of generality that: $\db_{H} = \db'_{C} = 1$ and $\db_{C} = \db'_{H} = 0$. Notice that this implies that $q(D)\neq q(D')$ for such $D$ and $D'$.

To finish the proof of Theorem \ref{th-2} it is now enough to show:

\begin{lemma}\label{lem:ucq-h10-is-ucq-det}
There exists a pair of different structures $\db, \db'$ over schema $\schema$ that satisfies $\views(\db) = \views(\db')$ if and only if $\hinstance$ has a solution over natural numbers.
\end{lemma}
\begin{proof}
($\Leftarrow$). Let $\hsolution$ be a solution over $\N$ of $\hinstance$ and let $a_i$ be a value of $x_i$ in $\hsolution$. Then let $\db$ and $\db'$ be such that:
\begin{itemize}
    \item $\db_H = 1$, $\db'_H = 0$, $\db_C = 0$, $\db'_C = 1$,
    \item $\db_{X_i} = \db'_{X_i} = a_i$,
\end{itemize}
From~\cref{lem:ucq-psi-p,lem:ucq-psi-n} we show that:
$$ V_\hinstance(\db) - V_\hinstance(\db') = \sum_{m \in P} m_\hsolution + \sum_{m \in N} m_\hsolution = \sum_{m \in \hinstance} m_\hsolution = 0$$

($\Rightarrow$). 
Now we will show that $\sum_{m \in \hinstance} m_D = 0$. From \cref{lem:ucq-psi-p,lem:ucq-psi-n,lem:ucq-counter-property} we get:
\begin{align*}
    V_\hinstance(\db) &= V_\hinstance(\db') \\
    \Psi_P(\db) + \Psi_N(\db) &= \Psi_P(\db') + \Psi_N(\db') \\
    \Psi_P(\db) &= \Psi_N(\db') \\
    \sum_{m \in P} m_D  &=  - \sum_{m \in N} m_D \\
    \sum_{m \in P} m_D  + \sum_{m \in N} m_D &= 0\\
    \sum_{m \in \hinstance} m_D  &= 0 
\end{align*}
\end{proof}


\section{Appendix B. Proof of Lemma \ref{lem:path-queries}($\Rightarrow$).}

Suppose there is no path, in $G_{q,V}$, from $\varepsilon$ to $q$.
We will show that in such case $V$ does not determine $q$. Let structure $D$ be defined in the following way:

\begin{itemize}
    \item $dom(D) = \{ [w,j] \mid w \text{ is a prefix of }q, j\in\{0,1\} \}$
    \item For $[w, i], [u, j] \in dom(D), R \in \schema$ we have $R([w,i], [u,j]) \in D$ if and only if 
$ u = wR$ and $i = j$.
\end{itemize}

\noindent
So $D$ is just $q + q$, that is the union of two disjoint frozen bodies of $q$.
It follows easily from the definition that  $\langle [\varepsilon,0], [q,0]\rangle\in q(D)$ with multiplicity 1.\\


For $w,u\in G_{q,V}$ we define $w\thicksim u$ if either both $w$ and $u$ are reachable, in graph $G_{q,V}$, from $\varepsilon$ or if none
of them is. Clearly,  $w\thicksim w'$  is an equivalence relation with two equivalence classes, and with $\varepsilon \not\thicksim q$.\\

We will now define the second structure, our $D'$, with the same domain as $D$, but different atoms:

For $u = wR$:
\begin{itemize} 
\item
$R([w,i], [u,i]) \in D'$ if and only if $u = wR$ and $w\thicksim u$;

\item if $i\neq j$ then $R([w,i], [u,j]) \in D'$ if and only if  $u = wR$ and $w\not\thicksim u$.
\end{itemize}

\noindent

Notice that this means that  if there is any  path in $D'$ from some $[w,i]$ to $[u,j]$ then:
\begin{center}
$i=j$ if and only if $w\thicksim u$.
\end{center}

Which in particular means that $\langle [\varepsilon,0], [q,0]\rangle\not\in q(D')$ and hence $q(D)\neq q(D')$.

But on the other hand, for $v \in V$, it is very easy to see that if $uvu'=q$ for some $u,u'\in \schema^*$ and if $i\in \{0,1\}$
then 
$\langle [u,i], [uv,i]\rangle\in v(D)$, with multiplicity 1, and that there are no other tuples in $V(D)$.
And it is also not hard to verify that such $V(D)$ exacty equals $V(D')$. This is since, if 
$u$ and $v$ are as above, then $u\thicksim uv$.


\section{Appendix C. Proofs of some lemmas needed for Theorem \ref{th-1}.}

\subsection{Proof of Lemma \ref{niby-thue}}
We only show the first claim, as the other one is symmetric.
The proof will be by induction with respect to $|w|$:

\noindent
(1) $|w| = |q|$. Then $w = q$ and we are done.\\
 (2)  $|w| > |q|$. Then there is $i$ such that $\iota_i = -1$.
  By \cref{def:walk} (1) we know that $\iota_1 = 1$, so there exists $j < i$ such that $\iota_{j} = 1$ and $\iota_{j+1} = -1$. Then, by \cref{def:walk} (3) we conclude that $A_{j} = A_{j+1}=A$ for some $A\in \schema$.
    This means that $w=uAA^{-1}u'$ for some $u,u'\in \bar\schema^*$. It is easy to see that the word $uu'$ constitutes a $q$-walk. And it is shorter than $w$. So, by the hypothesis, we have $uu' \thuetrans q$. And of course there 
    is also $w\thue uu'$, so we get $w \thuetrans q$. \hfill*$\square$

\subsection{Proof of Lemma \ref{cor:reductions}}
(1) If  $u\thue u'$ then there are $w,w'\in \bar\schema^*$ and $R\in \schema$ such that 
    $u=wRR^{-1}w'$ and $u'=ww'$. Then, using Lemma \ref{zawieranie-id}:\\\vspace{-2mm}
    
    $H_u=H_{wRR^{-1}w'} = H_{w'} H_R^{-1} H_R H_w =$\\
    \hspace*{4mm}$=H_{w'} \rel{h_R}^{-1} \rel{h_R} H_w \subseteq H_{w'} I H_{w} = H_{ww'}=H_{u'}$\\\vspace{-2mm}
    
    \noindent
 (2)   If  $u\theu u'$ then there are $w,w'\in \bar\schema^*$ and $R\in \schema$ such that
    $u=wR^{-1}Rw'$ and $u'=ww'$.  Then, again using Lemma \ref{zawieranie-id}: \\\vspace{-2mm}
    
    $H_u=H_{wR^{-1}Rw'} = H_{w'} H_R H_R^{-1} H_w =$\\
    \hspace*{4mm}$=H_{w'} \rel{h_R} \rel{h_R}^{-1}W H_w \supseteq H_{w'} I H_w = H_{ww'}=H_{u'}$ \hfill $\square$


\section{Appendix D.  Proofs of some lemmas needed for Theorem \ref{th-3}.}

\subsection{Proof of Lemma \ref{main-lemma-3} $(\Leftarrow) $.}\label{sec-main-th-2}

Let  $D, D'$ be some structures  such that $v(D) = v(D')$ for each  $v \in V_0$. 
We need to show that $q(D) = q(D')$. There are two cases:\\

\noindent
\textbullet~ {\bf Case 1:} $\exists v \in V \; v(D) = 0$.

    Then of course also $v(D') = 0$. By Observation \ref{najprostsza} this implies that $q(D) = 0$. Likewise we get that  $q(D') = 0$, so $q(D) = q(D')$.\\
    
    \noindent
\textbullet~ {\bf Case 2:} $\forall v \in V ~ v(D) \neq 0$.

    Take $\alpha_1, ..., \alpha_k \in \R$ such that $\vec q = \sum_{i = 1}^{|V|} \alpha \vec v_i$.
    \begin{align*}
        q(D) &= \prod_{i=1}^k w_i(D)^{\vec q(i)} & & \text{(by \cref{obs:query-eval})}\\
        &= \prod_{i=1}^k w_i(D)^{\sum_{j=1}^{|V|} \alpha_j \vec v_j(i)} \\
        &= \prod_{i=1}^k \prod_{j=1}^{|V|} \left(w_i(D)^{\vec v_j(i)} \right)^{\alpha_j} \\
        &= \prod_{j=1}^{|V|} \left(\prod_{i=1}^k w_i(D)^{\vec v_j(i)} \right)^{\alpha_j} \\
        &= \prod_{j=1}^{|V|} v_j(D)^{\alpha_j} & & \text{(by \cref{obs:query-eval} again)}
    \end{align*}

Note that since for every $j \in \{ 1, ..., k \}$ we have $ v_j(D) > 0$, the expression $\prod_{j=1}^{|V|} v_j(D)^{\alpha_j}$ is correct, even if for some $j$ the number $\alpha_j$ is not  natural.

Likewise, we show that $q(D') = \prod_{j=1}^{|V|} v_j(D')^{\alpha_j}$. However, we know that for $j \in \{ 1, ..., k \}$ 
it holds that $v_j(D) = v_j(D')$, so this implies that $q(D') = \prod_{j=1}^{|V|} v_j(D)^{\alpha_j} = q(D)$. \hfill $\square$

\subsection{Proof of Observation \ref{obs-krok-1} }
\noindent
 By Lemma \ref{lem:hom-properties} we have:\vspace{-2mm}
$$|hom(w,s^{(2)})| = |hom(w, \sum_{i=1}^m T^i s_i^{(1)})| = \sum_{i=1}^m T^i |hom(w, s_i^{(1)})|$$ 
Likewise $ |hom(w',s^{(2)})| = \sum_{i=1}^m T^i |hom(w',s_i^{(1)})|$. 

Notice that this means that the sequence:
$$|hom(w, s_m^{(1)})|;|hom(w, s_{m-1}^{(1)})|;\ldots |hom(w, s_1^{(1)})|;0 $$ 
is a representation\footnote{Recall that each of $|hom(w, s_i^{(1)})|$ is smaller than $T$.}  of  $|hom(w,s^{(2)})|$ in radix $T$.
And:
$$|hom(w', s_m^{(1)})|;|hom(w', s_{m-1}^{(1)})|;\ldots |hom(w', s_1^{(1)})|;0 $$ 
\noindent
is a representation of  $|hom(w',s^{(2)})|$ in radix $T$. The two representations are different since 
$|hom(w, s_i^{(1)})|\neq$ $|hom(w', s_i^{(1)})|$ for some $i$ (by Step 1). So the two numbers must be different too. \hfill$\square$

\subsection{Proof of Lemma \ref{ilepierwiastkow}}


 Take any $\alpha_1, ..., \alpha_k \in \R$ such that for all $i \in \range{k}$, $\sum_{j=1}^{k} \alpha_j A(i, j) = 0$. 
We will show that then $\alpha_1 = \dots = \alpha_k = 0$. Let us define a polynomial $P(X) = \alpha_1 + \alpha_2 X + ... + \alpha_k X^{k-1}$. Then $\sum_{j=1}^{k} \alpha_j A(i, j) = \sum_{j=1}^{k} \alpha_j a_i^{j-1} = P(a_i)$. Because $a_1, ..., a_k$ are pairwise distinct, we know that $P$ has at least $k$ zeros. But the degree of $P$ is at most $k-1$, hence $P = 0$ and all of its coefficients are $0$.\hfill$\square$

\subsection{Proof of Observation \ref{obs:vecpow}}
\begin{enumerate}
    \item
    \begin{align*}
        \vecpow{(\vec u \circ \vec v)}{\vec w} &= \prod_{i=1}^k (\vec u(i) \vec v(i))^{\vec w(i)} \\
        &= \prod_{i=1}^k (\vec u(i))^{\vec w(i)} \prod_{i=1}^k (\vec v(i))^{\vec w(i)} \\
        &= (\vecpow{\vec u}{\vec w})(\vecpow{\vec v}{\vec w})
    \end{align*}
    \item 
    \begin{align*}
    \vecpow{t^{\vec u}}{\vec v} &= \prod_{i=1}^k \left(t^{\vec u(i)} \right)^{\vec v(i)}
    = \prod_{i=1}^k t^{\vec u(i) \vec v(i)} \\
    &= t^{\sum_{i=1}^k \vec u(i) \vec v(i)}
    = t^{\dotprod{\vec u}{\vec v}}
    \end{align*}
\end{enumerate}

\subsection{Proof of Lemma \ref{lem:eval}}

\begin{align*}
(M \vec s)(i) &= \sum_{j=1}^k M(i, j)\vec s(j) \\
&= \sum_{j=1}^k w_i(s_j) \vec s(j) & \text{(by definition of $M$)} \\
&= w_i\left( \sum_{j=1}^k \vec s(j) s_j \right) & \text{(by \cref{lem:hom-properties})} \\
&= w_i(s) & \text{(by definition of $\vec s$)}
\end{align*}
\begin{align*}
    \vecpow{(M \vec s)}{\vec v} &= \prod_{i=1}^k (M \vec s)(i)^{\vec v(i)} \\
    &= \prod_{i=1}^k w_i(s)^{\vec v(i)} \\
    &= \left( \sum_{i=1}^k \vec v(i) w_i \right)(s) & \text{(by \cref{lem:hom-properties})} \\
\end{align*}
\hfill $\square$